\documentclass[a4paper, 10pt, conference]{ieeeconf}      
\usepackage{eso-pic}
\newcommand\AtPageUpperMyright[1]{\AtPageUpperLeft{%
		\put(\LenToUnit{0.5\paperwidth},\LenToUnit{-1cm}){%
			\parbox{0.5\textwidth}{\raggedleft\fontsize{9}{11}\selectfont #1}}%
}}%
\newcommand{\conf}[1]{%
	\AddToShipoutPictureBG*{%
		\AtPageUpperMyright{#1}
	}
}

\usepackage{color}
\usepackage{harvard}

\usepackage{graphicx}          
\IEEEoverridecommandlockouts                              

\usepackage[T1]{fontenc}
\usepackage[utf8x]{inputenc}

\usepackage{amsfonts}
\usepackage{amsmath}
\usepackage{amssymb}
\usepackage{bbm}
\usepackage{amstext,framed,pdfsync}
\usepackage{amssymb}
\usepackage{subcaption}
\usepackage{cite}

\newcommand{\tx}[1]{\text{#1}}

\newcommand{\vx}[1]{\mathbf{#1}}

\newtheorem{theorem}{Theorem}[section]
\newtheorem{corollary}[theorem]{Corollary}
\newtheorem{proposition}[theorem]{Proposition}
\newtheorem{remark}[theorem]{Remark}

\usepackage{hyperref}

\title{\LARGE \bf
	Integral Line-of-Sight Curved Path Following of Helical Microswimmers Actuated by 
	Rotating Magnetic Dipoles
}                                              

\author{Alireza Mohammadi$^{1,\ast}$ and Mark W. Spong$^{2}$
	\thanks{$^{1}$Alireza Mohammadi is with the Department of Electrical
		 \& Computer Engineering, University of Michigan, Dearborn.
		{\tt\small amohmmad@umich.edu}}%
	\thanks{$^{2}$Mark W. Spong is with the Erik Jonsson School of Engineering 
		\& Computer Science, University of Texas, Dallas.
		{\tt\small mspong@utdallas.edu}}%
	\thanks{$^\ast$Corresponding author: Alireza Mohammadi}
}

\begin{document}
	
	\maketitle
	\conf{International Conference on Manipulation, Automation and Robotics at Small Scales (MARSS), Toronto, Canada 25--29 July 2022}

	\thispagestyle{empty}
	\pagestyle{empty}
		
	\begin{abstract}                          
		This short paper investigates the problem of curved path following for helical microswimmers actuated by rotating magnetic dipoles.  The proposed solution, which relies on an integral line-of-sight (ILOS) guidance law, can be utilized in both below and beyond step-out frequency regimes. 
	\end{abstract}
\section{Introduction}
%
%
The underlying mechanism of propulsion for helical  magnetic microswimmers is the 
transduction of magnetic torque to mechanical power~\cite{zhang2009artificial}.  For 
\emph{in vivo} applications, the bacterial flagella-inspired morphology  of 
these microswimmers is argued to provide the best 
overall choice~\cite{abbott2009should}. 

In our prior work in~\cite{mohammadi2021integral}, we presented a straight-line path following control law that formally 
guarantees practical convergence of magnetic microswimmers to desired straight lines.  The solution relies on utilizing an integral line-of-sight 
(ILOS)-based reference vector field adopted from a popular family of guidance laws for marine surface vessels~\cite{breivik2004path,caharija2016integral}.  Our path following control scheme employed    
an optimal decision strategy (ODS)-based control synthesis approach to guarantee that the microswimmer always respects the step-out frequency constraint. However, the dynamics of the microswimmer beyond the step-out frequency regime were not considered in that work. Furthermore, it was assumed that the microswimmer was moving under  a \emph{uniform
 rotating magnetic field} generated by a proper actuator such as an electromagnetic coil-based system presented in~\cite{zhang2009artificial}. 

\noindent\textbf{Contributions of the paper.} This short paper extends our prior work in~\cite{mohammadi2021integral} in several important ways. First, we assume that the microswimmer is actuated by a rotating magnetic dipole in contrast to a  uniform field (see, e.g.,~\cite{wright2017spherical,chaluvadi2020kinematic} for such magnetic dipole-based actuation systems). Second, we present the non-smooth dynamics of microswimmers beyond the step-out frequency regime and provide conditions under which the microswimmer trajectories  vary continuously with respect to the swimmer's initial position. Finally, we modify our previously proposed ILOS-based guidance law to ensure convergence to \emph{curved} planar paths. 

The rest of this short paper is organized as follows. First, we present the unified non-smooth dynamical model of 
swimming helical microrobots below and beyond step-out frequencies in Section~\ref{sec:dynmodel}. Next, we present 
the main results in Section~\ref{sec:ContProb}. After  presenting the simulations in Section~\ref{sec:sims}, we 
conclude the paper in Section~\ref{sec:conc}.

\noindent{\textbf{Notation.}} We let $\mathbb{R}_{+}$ denote the set of all non-negative 
real numbers.  Given $\mathbf{x}\in\mathbb{R}^N$, we 
let $|\mathbf{x}|:=\sqrt{\mathbf{x}^{\top}\mathbf{x}}$ 
denote the Euclidean norm of $\mathbf{x}$, and if $\mathbf{x}\neq \mathbf{0}$, then we let the unit-normalized vector parallel 
to $\mathbf{x}$ be given by $\hat{\mathbf{x}}:=\tfrac{\mathbf{x}}{|\mathbf{x}|}$.  We let $\mathbbm{1}:\mathbb{R}\to \{0,1\}$ denote the Heaviside step function, 
where $\mathbbm{1}(x)=1$ if $x\geq 0$, and $\mathbbm{1}(x)=0$ if $x<0$. Given the integer $N\geq 1$, we let $\mathbf{I}_N$ 
denote the identity matrix  of order $N$.
%

\section{Dynamical Model of Magnetic Helical Microswimmers}\label{sec:dynmodel} 
In this section we present the non-smooth dynamics of helical microswimmers under 
rotating magnetic dipole fields. 

We consider the planar curve $\mathcal{C}$ parameterized by a scalar variable $\tau\in\mathbb{R}$ to which we would like the microswimmer to converge. For any given $\tau$, the planar path is denoted by $\mathbf{p}_d(\tau)\in \mathbb{R}^2$. For brevity, we assume that the gravitational acceleration vector $\mathbf{g}$ and $\mathcal{C}$ are in the same plane. Considering  $\vx{g}$, we fix the right-handed inertial coordinate frame  with axes $X_I$ and $Y_I$, and its origin located at the dipole center as depicted in Figure~\ref{fig:micro}. We let $\mathbf{p}=[p_x,\,p_y]^\top$ denote the microswimmer position vector. 
\begin{figure}[t!]
	\centering
	%
	%
    	\includegraphics[scale=0.32]{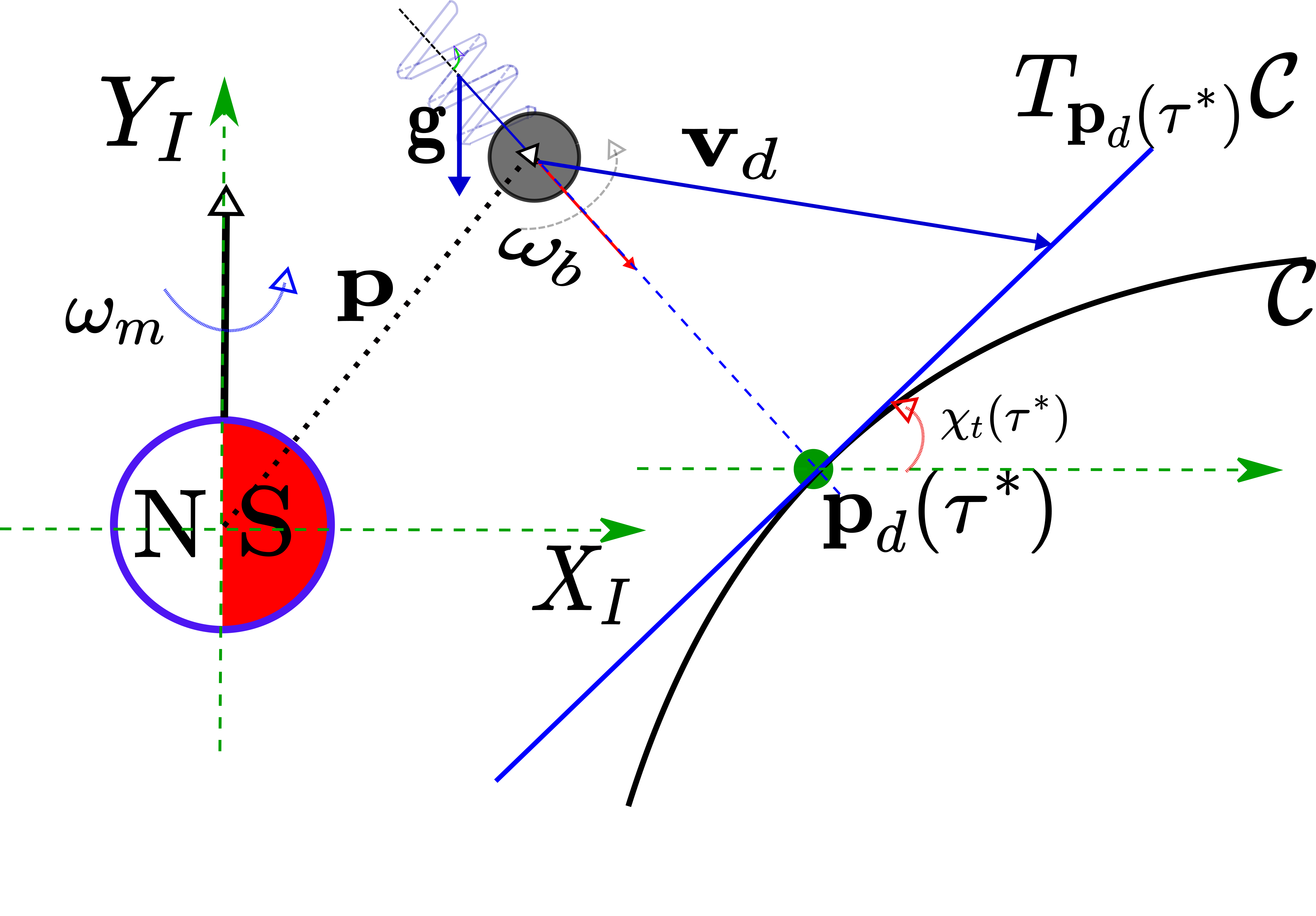}
    	\vspace{-1ex}
	%
	%
	\caption{\small The configuration of the helical microswimmer and the rotating magnetic dipole.}
	\vspace{-5ex}
	\label{fig:micro}
\end{figure}
It is assumed that a magnetic dipole located at the origin of the inertial coordinate frame is rotated about, and orthogonal  to, 
the vector $\pmb{\omega}_m$ parallel to $Y_I$. Such a rotating magnetic-dipole-field source has been implemented in  the SAMM system presented in~\cite{wright2017spherical}. 

Assuming  a fluidic environment of low-Reynolds-number and  rotational synchrony in-between the microswimmer and 
the dipole-induced magnetic field, the swimmer dynamics are governed by (see, e.g.,~\cite{mohammadi2021integral}
and~\cite{mahoney2014behavior})
\begin{equation}\label{eq:eq_dyn}
\dot{\vx{p}} = \mathbf{w} + {\vx{d}}_{\mu}, 
\end{equation}
where $\mathbf{d}_\mu$ is the vector of disturbances, and the dipole-induced microswimmer velocity vector $\mathbf{w}$
is given by~\cite{chaluvadi2020kinematic}
\begin{equation}\label{eq:eq_dyn2}
\mathbf{w}=\frac{\Gamma^{-1}(\vx{p}) \pmb{\omega}_m}{|\Gamma^{-1}(\vx{p}) \pmb{\omega}_m|}F(|\pmb{\omega}_m|), 
\end{equation}
in which the invertible matrix $\Gamma({\vx{p}}) = 3\frac{\vx{p}\vx{p}^\top}{|\vx{p}|^2} - \mathbf{I}_2$, with inverse $\Gamma^{-1}(\vx{p}) = \frac{3}{2}\frac{\vx{p}\vx{p}^\top}{|\vx{p}|^2} - \mathbf{I}_2$, captures the shape of the dipole field. Moreover, the vector $\pmb{\omega}_m$, which is the angular velocity of rotation of the dipole, can be considered as the control input for~\eqref{eq:eq_dyn}. In~\eqref{eq:eq_dyn2}, the step-out frequency effect is captured by the non-smooth function $F:\mathbb{R}\to\mathbb{R}$ given by 
\begin{equation}\label{eq:eq_dyn3}
F(x) = \beta \big(x - \mathbbm{1}(x^2 - \Omega_{\tx{SO}}^2)\sqrt{x^2 - \Omega_{\tx{SO}}^2}\big), 
\end{equation}
\noindent where $\beta$ is a constant depending on the physical/geometrical parameters of the microswimmer and its fluidic ambient environment, and  $\Omega_{\tx{SO}}$ is the step-out frequency beyond which the microswimmer speed rapidly
declines~\cite{mahoney2014behavior}. 

\section{Main Results}
\label{sec:ContProb}

\subsection{Continuity with respect to the initial conditions}
The dynamics of the microswimmer under a rotating dipole field in~\eqref{eq:eq_dyn}--\eqref{eq:eq_dyn3} while considering the step-out frequency effect, as captured by $F(\cdot)$ in~\eqref{eq:eq_dyn3}, are non-smooth and non-Lipschitz continuous. 

Any feedback control scheme designed for the  non-smooth dynamics in~\eqref{eq:eq_dyn}--\eqref{eq:eq_dyn3} must guarantee continuity of the solutions with respect to initial conditions for the following \emph{practical reasons}. First, there always exists difficulty in precise localization of microswimmers, especially in clinical use. Moreover, the family of centroid/covariance-based control schemes for swarms of microswimmers rely on applying a single global magnetic field to the whole swarm using the centroid position information~\cite{chaluvadi2020kinematic}. Continuity with respect to the microswimmer initial position ensures that the resulting microswimmer trajectories remain close to their nominal counterparts even if the applied control input does not incorporate the exact microswimmer position information. The following proposition provides a sufficient condition for ensuring the required continuity. 
\begin{proposition}
	Consider the microswimmer dynamics given by~\eqref{eq:eq_dyn}--\eqref{eq:eq_dyn3} under the feedback control input $\pmb{\omega}_m=\pmb{\omega}_m(\mathbf{p})$. If $\pmb{\omega}_m(\cdot)$  is a continuous function of $\mathbf{p}$, then  the resulting 
	closed-loop dynamics admit solutions that are continuous with respect to the initial microswimmer positions.
\end{proposition}
\begin{proof}
	For brevity, we provide a sketch of the proof. A continuous closed-loop feedback input makes the resulting dynamics one-sided Lipschitz. Therefore, the closed-loop dynamics admit unique solutions (see, e.g., Proposition~2 in~\cite{cortes2008discontinuous}). Being a continuous closed-loop vector field and having a unique solution, the statement of the proposition follows from Theorem 4.3 in~\cite{coddington1955theory} (see p. 59). 
\end{proof}
\begin{remark}
	Continuity  of solutions with respect to initial conditions are also necessary for stability-based 
	switching schemes for nonlinear system stabilization (see, e.g.,~\cite{leonessa1998nonlinear}).  
\end{remark}
\subsection{Guidance law for curved path following}
Considering the curved path $\mathcal{C}$ in Section~\ref{sec:dynmodel} and the microswimmer closed-loop 
dynamics in~\eqref{eq:eq_dyn}--\eqref{eq:eq_dyn3}, 
we propose an ILOS-based guidance law that is adopted from the marine 
surface vessel control literature (see, 
e.g.,~\cite{breivik2004path,caharija2016integral}). 

Following~\cite{breivik2004path}, it is assumed that at each $\mathbf{p}$, there exists a well-defined value of $\tau$  that minimizes the Euclidean distance between $\mathbf{p}$ and $\mathbf{p}_d(\tau)$, given by $\tau^\ast=\text{argmin}_{\tau\in\mathbb{R}} |\mathbf{p} - \mathbf{p}_d(\tau)|$.  Given the microswimmer position vector  $\mathbf{p}$ and its associated nearest point $\mathbf{p}_d(\tau^\ast)$ on the path, we denote the line that is tangent to $\mathcal{C}$ at $\mathbf{p}_d(\tau^\ast)$ by $T_{\mathbf{p}_d(\tau^\ast)}\mathcal{C}$ and the tangent line slope by $\chi_t(\tau^\ast)$ (see Figure~\ref{fig:micro}). The \emph{control objective} is to make the microswimmer to converge to and move on $\mathcal{C}$.  We propose the following guidance law, which is based on guiding the microswimmer from  its position $\mathbf{p}$ to the tangent line to $\mathcal{C}$ at the nearest point $\mathbf{p}_d(\tau^\ast)$ on $\mathcal{C}$ (see  Figure~\ref{fig:micro}), 
\begin{equation}\label{eq:los}
\vx{v}^{\tx{d}}(\vx{p},s) = \alpha_\text{d}\vx{R}_{\chi_t(\tau^\ast)} \begin{bmatrix} \Delta_\tx{LOS} \\ -| \vx{p}-\mathbf{p}_d(\tau^\ast)| \sin(\Delta \theta) - 
\sigma_0 s \end{bmatrix},
\end{equation}
where $\Delta \theta:=\text{atan2}(p_y-y_d(\tau^\ast),p_x-x_d(\tau^\ast))-\chi_t(\tau^\ast)$, and  the matrix $\mathbf{R}_{\chi_t(\tau^\ast)}$ is the rotation matrix by the slope angle $\chi_t(\tau^\ast)$. In~\eqref{eq:los}, the design parameters $\alpha_\text{d}>0$ and the integral gain $\sigma_0>0$ are constant. Additionally, the parameter $\Delta_\tx{LOS}>0$, which determines the point along 
$T_{\mathbf{p}_d(\tau^\ast)}\mathcal{C}$ at which the microrobot should be pointed, is called 
the \emph{look-ahead distance}. Furthermore, the variable $s$ is governed by  
\begin{equation}\label{eq:losdynvar}
\dot{s} = -k_\text{d}s + 
\Delta_{\tx{LOS}}\frac{| \vx{p}-\mathbf{p}_d(\tau^\ast)| \sin(\Delta \theta)}{(| \vx{p} -\mathbf{p}_d(\tau^\ast)| 
	\sin(\Delta \theta)+\sigma_0 s)^2 + \Delta_{\tx{LOS}}^2},
\end{equation}
\noindent where the damping gain $k_\text{d}>0$ is a constant design parameter.  In the presence of disturbances that drive the
microswimmer away from its desired path, embedding an integral control action through the  dynamic variable $s$  builds up a 
corrective compensation in the reference vector field. 

Considering the ILOS-based guidance vector field~\eqref{eq:los} and~\eqref{eq:losdynvar}, we propose using the control law 
\begin{equation}
\pmb{\omega}_m = \Gamma(\vx{p}) \vx{v}^{\tx{d}}(\vx{p},s),
\label{eq:contrLaw}
\end{equation}
for the microswimmer dynamics in~\eqref{eq:eq_dyn}--\eqref{eq:eq_dyn3}. Under this control law, the step-out frequency constraint 
$\pmb{\omega}_m^\top \pmb{\omega}_m \leq \Omega_{\text{SO}}^2$ is not necessarily respected. Similar to the arguments in~\cite{mohammadi2021integral}, it can be shown that the solution to the quadratic programming problem
\begin{equation}\label{eq:ODS}
\begin{aligned}
\underset{\pmb{\omega}_m}{\text{min.}} \; & \big\{ \frac{1}{2} \big(\Gamma(\mathbf{p})\pmb{\omega}_m\big)^{\top} \mathbf{A}_{\mu} \big(\Gamma(\mathbf{p})\pmb{\omega}_m\big) + \vx{G}^{\top}_{\mu}(\vx{p},s) \Gamma(\mathbf{p})\pmb{\omega}_m \big\} \\
\text{subject to} & \;\;\; \pmb{\omega}_m^{\top} \pmb{\omega}_m \leq \Omega_{\text{SO}}^2, 
\end{aligned}
\end{equation}
where
\begin{equation}\label{eq:ODS2}
\textcolor{black}{
	\begin{aligned}
	\vx{A}_{\mu} := \frac{1}{\Omega_0^2} \mathbf{I}_2, \\
	\vx{G}_{\mu}(\vx{p},s) := \frac{-\vx{A}_{\mu}}{\beta} \big( \vx{d}_{\mu} -   \vx{v}^{\tx{d}}(\vx{p},s) \big), 
	\end{aligned}
}
\end{equation}
for some $\Omega_0>0$ respects the step-out frequency constraint. 

\section{Simulation Results}
\label{sec:sims}
In this section we present numerical simulation results to validate 
the performance of the proposed control method described in 
Section~\ref{sec:ContProb}. The physical parameters of the microswimmer are chosen according 
to~\cite{mahoney2011velocity} (also, see Tables~1 and~2 in~\cite{mohammadi2021integral}). The path $\mathcal{C}$ is a circle of radius $1.5$ cm centered at the origin.  The controller parameters are chosen to be $\alpha_d=0.01$, $k_d=0.15$, and $\Delta_{\tx{LOS}}=0.75$ mm. The step-out frequency is chosen to be $\Omega_{\text{SO}}=2\pi f_{\text{SO}}$, with  $f_{\text{SO}}=2.8$ Hz. We perform the numerical simulations under two control scenarios. In the first scenario (see Figures~\ref{fig:path_nominal}(a, c)), we utilize the control input in~\eqref{eq:contrLaw}. In the second scenario (see Figures~\ref{fig:path_nominal}(b, d)), we solve the quadratic programming problem in~\eqref{eq:ODS} with  $\Omega_0=2\pi \tfrac{\text{rad}}{\text{s}}$, to which a closed-form solution exists (see~\cite{mohammadi2021integral} for further details), for generating a proper control input $\pmb{\omega}_m$ respecting  the step-out frequency constraint. In Figure~\ref{fig:path_nominal}, the justification for the slower speed in the first scenario is that the microswimmer is allowed to rotate faster than the magnetic field step-out frequency. Therefore,  its speed decreases with respect to the second scenario where the controller respects the step-out frequency constraint. 
%
%
%
\begin{figure}
	\centering
	\begin{subfigure}{0.32\textwidth}
		\includegraphics[width=0.95\textwidth]{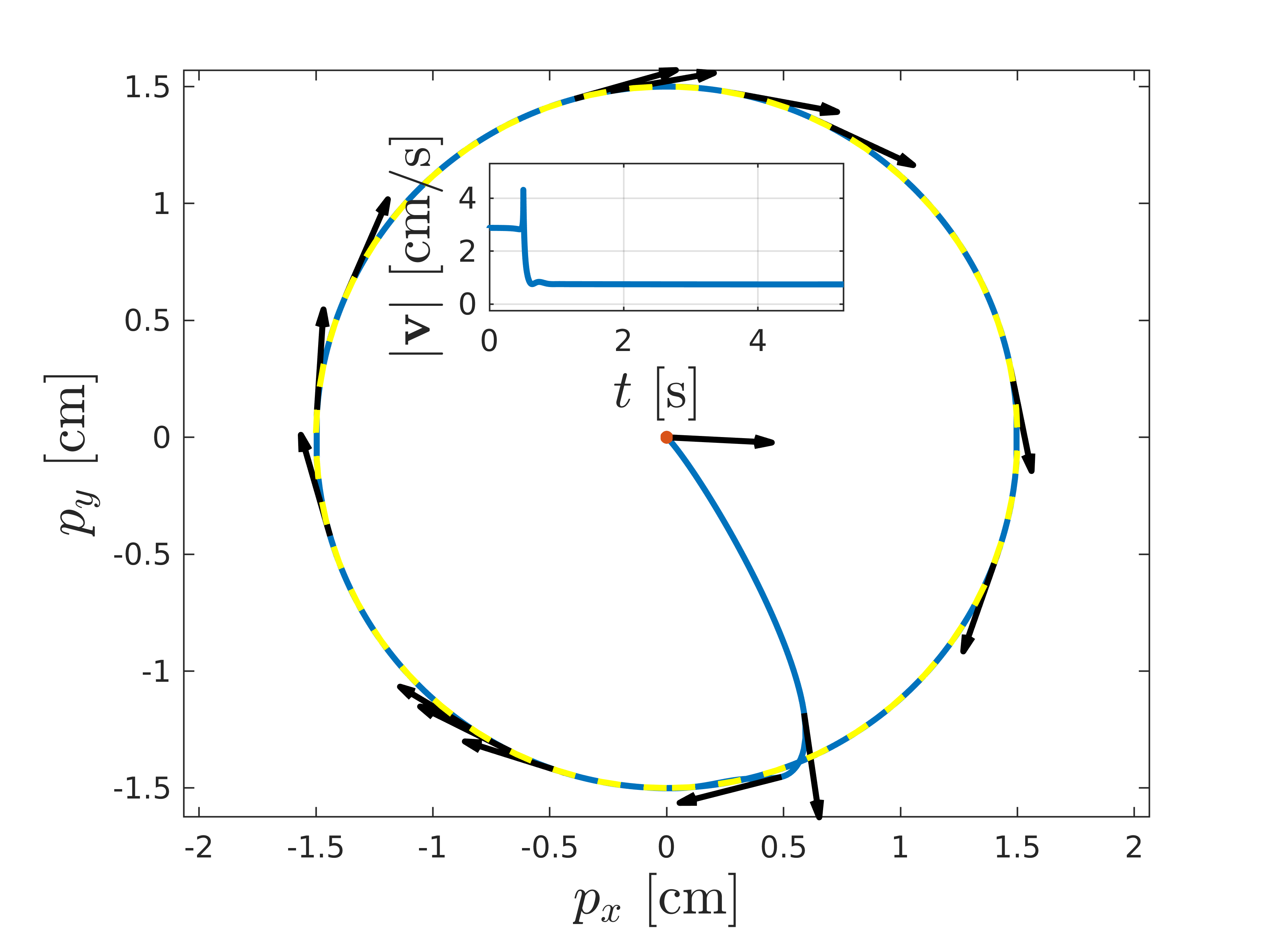} 
		\caption{}
	\end{subfigure}
    \hspace{-1ex}
	\begin{subfigure}{0.32\textwidth}
		\includegraphics[width=0.95\textwidth]{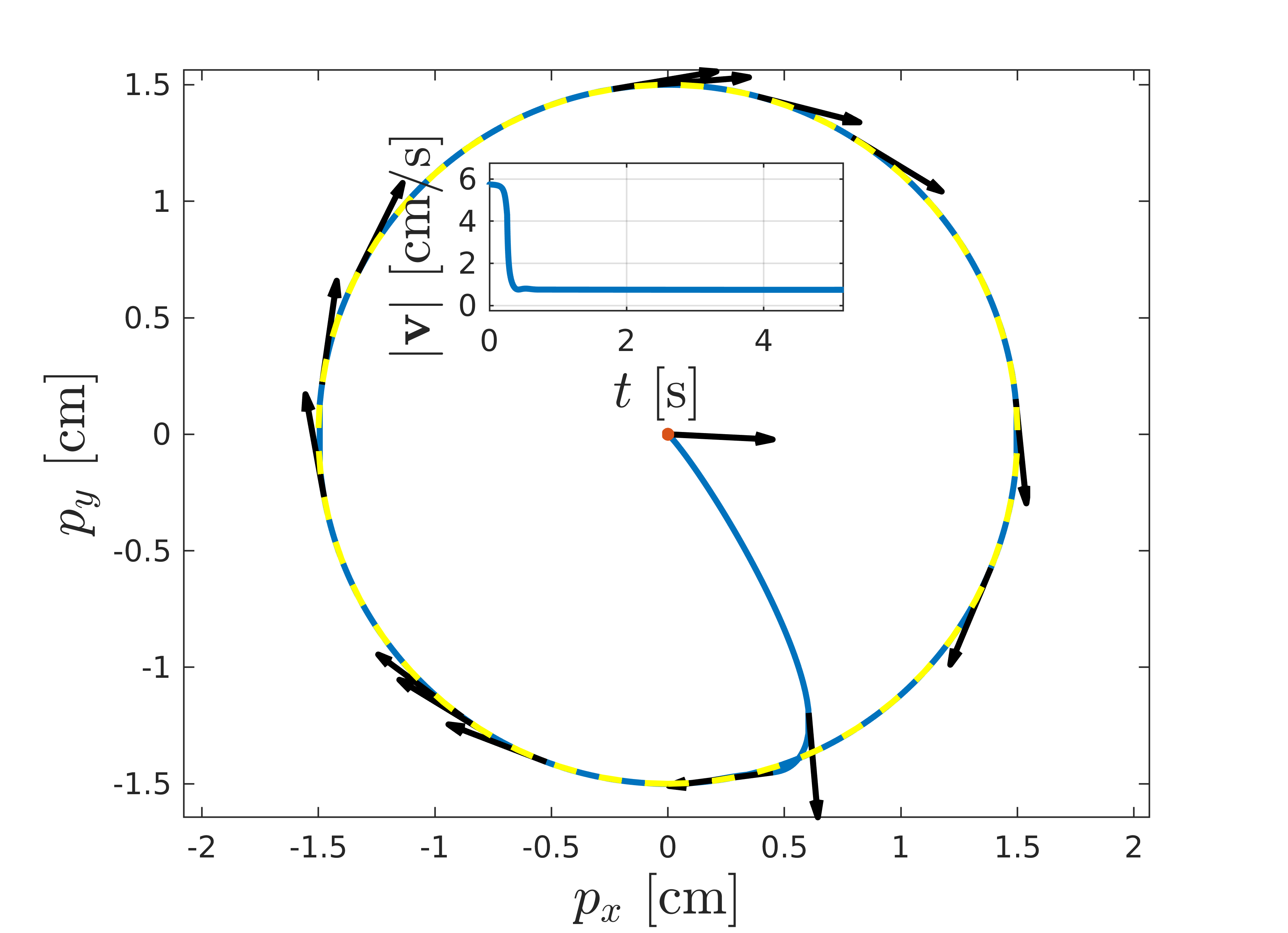} 
		\caption{}
	\end{subfigure}
	    \hspace{-1ex}
	\begin{subfigure}{0.32\textwidth}
		\includegraphics[width=0.85\textwidth]{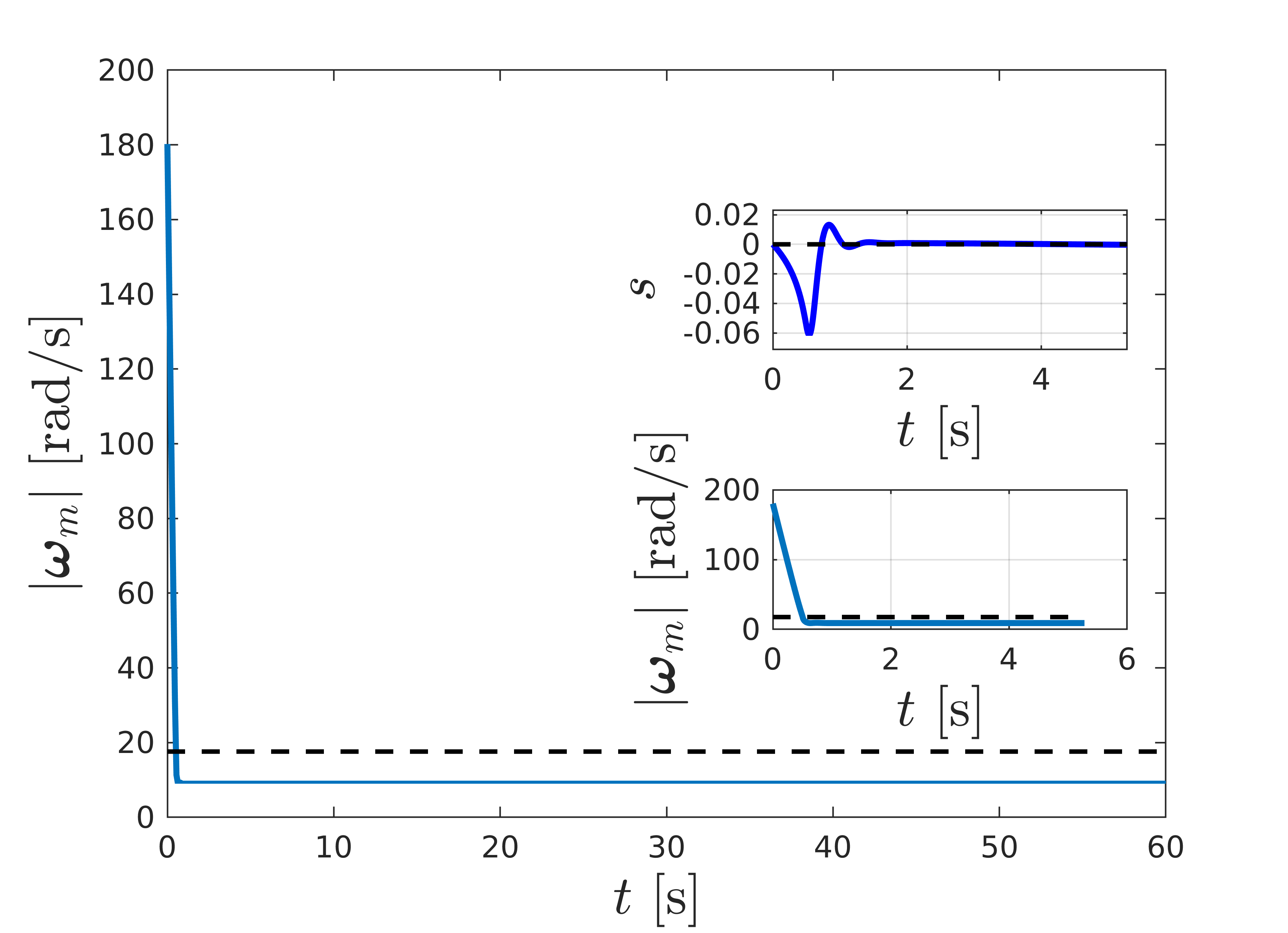} 
		\caption{}
	\end{subfigure}
    \hspace{-1ex}
	\begin{subfigure}{0.32\textwidth}
		\includegraphics[width=0.85\textwidth]{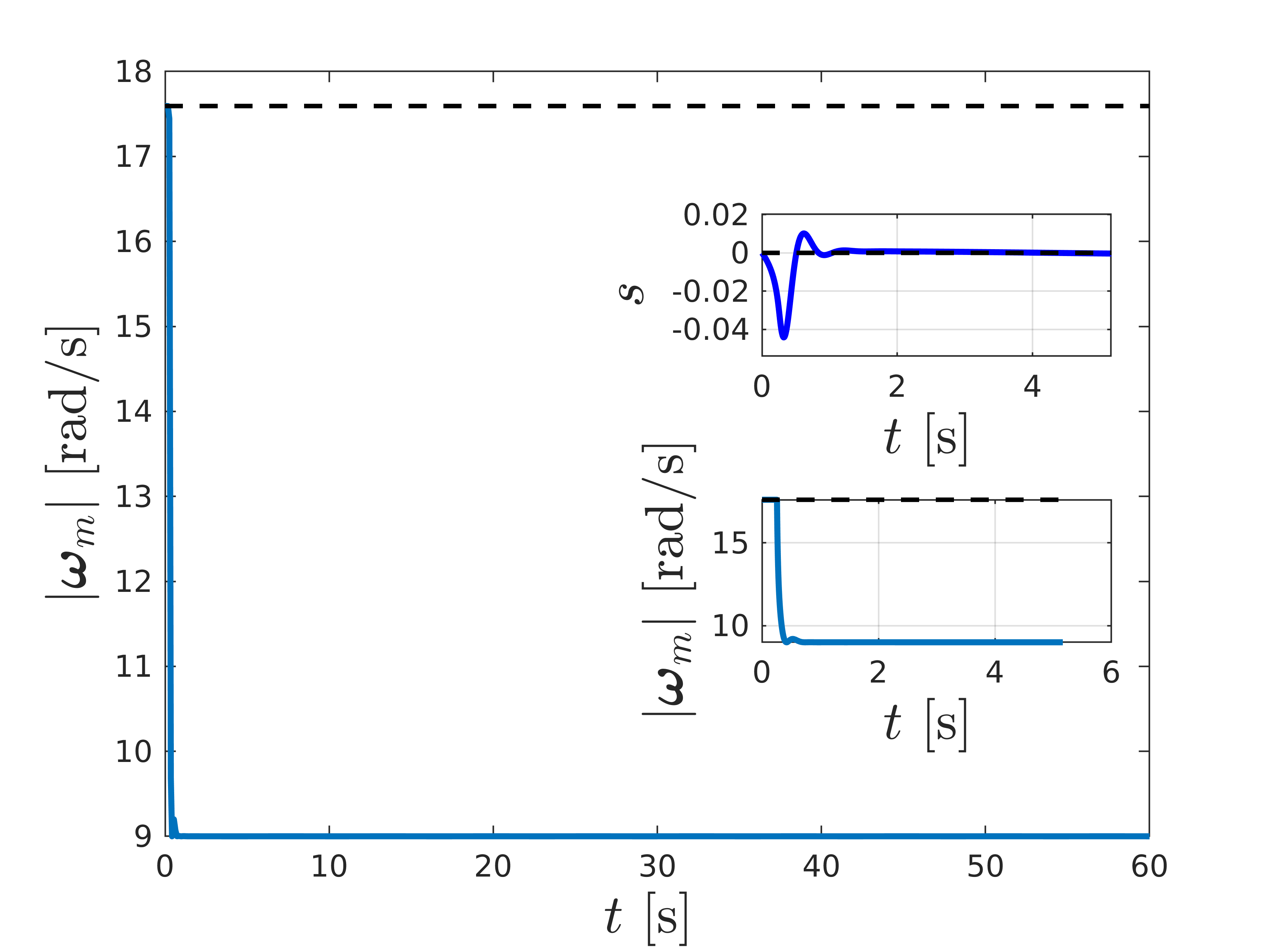} 
		\caption{}
	\end{subfigure}
	\caption{Microswimmer position, the time profile of the 
		velocity magnitude, the control input magnitude time profile, and the 
		time profile of the adaptive variable $s$: 
		(a, c)  not enforcing the step-out frequency constraint, and 
		(b, d)  enforcing the step-out frequency constraint.}
	\vspace{-4ex}
	\label{fig:path_nominal}
\end{figure}

\section{Concluding Remarks}
\label{sec:conc}
This short paper extends the straight-line path following control scheme of~\cite{mohammadi2021integral} to solving the problem of curved path following for helical microswimmers actuated by rotating magnetic dipoles.  Future research will include extending the presented results to control of microswimmer swarms under a global magnetic dipole field. 

\bibliographystyle{ieeetran}        
\bibliography{PhDBiblio}           

\end{document}